\documentclass[letterpaper, 11pt]{article}
\usepackage[utf8]{inputenc}
\usepackage[margin=1in]{geometry}
\usepackage[linesnumbered,vlined]{algorithm2e}
\usepackage{amsmath}
\usepackage{amsfonts, amssymb}
\usepackage{cite}
\usepackage{xcolor}
\usepackage{amsthm}
\usepackage{mathtools}
\usepackage{hyperref}
\usepackage[linesnumbered,vlined]{algorithm2e}
\usepackage{subfigure}

\usepackage{graphicx}

\usepackage{thmtools,thm-restate}
\newtheorem{theorem}{Theorem} 
\newtheorem{lemma}[theorem]{Lemma}

\usepackage{mathrsfs}

\newcounter{cntLemmaNumber}

\newcounter{cntTheoremNumber}

\newcommand{\eps}{\epsilon}

\newcommand{\size}[1]{\ensuremath{\left|#1\right|}}

\newcommand{\set}[1]{\left\{ #1 \right\}}

\DeclarePairedDelimiterX{\infdivx}[2]{(}{)}{%
  #1\;\delimsize\|\;#2%
}

\DeclarePairedDelimiter{\ceil}{\lceil}{\rceil}

\newif\ifarxiv
\arxivtrue

\begin{document}
\begin{titlepage}
\title{Local Max-Cut on Sparse Graphs}

\date{}
\ifarxiv
\author{Gregory Schwartzman\footnotemark[3]}
\renewcommand*{\thefootnote}{\fnsymbol{footnote}}
\footnotetext[3]{%
JAIST,
\texttt{greg@jaist.ac.jp}.
}
\fi
\end{titlepage}

\maketitle
\begin{abstract}
We bound the smoothed running time of the FLIP algorithm for local Max-Cut as a function of $\alpha$, the arboricity of the input graph. We show that, with high probability and in expectation, the following holds (where $n$ is the number of nodes and $\phi$ is the smoothing parameter):
\begin{enumerate}
    \item When $\alpha = O(\log^{1-\delta} n)$ FLIP terminates in $\phi poly(n)$ iterations, where $\delta \in (0,1]$ is an arbitrarily small constant. Previous to our results the only graph families for which FLIP was known to achieve a smoothed polynomial running time were complete graphs and graphs with logarithmic maximum degree.
    \item For arbitrary values of $\alpha$ we get a running time of $\phi n^{O(\frac{\alpha}{\log n} + \log \alpha)}$. This improves over the best known running time for general graphs of $\phi n^{O(\sqrt{ \log n })}$ for $\alpha = o(\log^{1.5} n)$. Specifically, when $\alpha = O(\log n)$ we get a significantly faster running time of $\phi n^{O(\log \log n)}$.
\end{enumerate}



\end{abstract}
\pagenumbering{gobble}
\newpage
\pagenumbering{arabic}

\section{Introduction}

\textbf{Local Max-Cut} Let us start by introducing the problem of local Max-Cut and the FLIP algorithm. Our input is a weighted graph $G(V,E)$, with $\size{V}=n, \size{E}=m$. The local Max-Cut problem asks us to find a cut ($C\subset V$) that is \emph{locally maximal} -- where the weight of the cut cannot be improved by moving any one vertex to the other side of the cut (throughout the paper we use the terms ``move" and ``flip" interchangeably). This is a very natural problem in the context of \emph{local search}, where it is known to be PLS-complete (PLS is the complexity class Polynomial Local Search) \cite{SchafferY91}. It also appears in the context of the party affiliation game \cite{FabrikantPT04} where it corresponds to a Nash equilibrium, and in the context of finding a stable configuration of Hopfield networks \cite{hopfield1982neural}.

\paragraph{FLIP} A simple algorithm for the problem is the following local search approach, called FLIP. Starting from any initial cut, pick an arbitrary vertex and move it to the other side of the cut as long as this results in an improvement to the weight of the cut. It is clear that once FLIP terminates, we have arrived at a locally maximal cut. It is known that there exist weight assignments to the edges such that this algorithm takes exponential time to terminate \cite{JohnsonPY88}. Nevertheless, empirical evidence suggests that the algorithm terminates in reasonable time \cite{AngelBPW17}. To bridge this gap, we are interested in the \emph{smoothed complexity} of the problem.


\paragraph{Smoothed analysis}  Smoothed analysis was introduced by \cite{SpielmanT04} in an attempt to explain the fast running time of the Simplex algorithm in practice, despite its exponential worst-case runtime. In the smoothed analysis framework, an algorithm is provided with an adversarial input that is perturbed by some random noise. We are then interested in bounding the running time of the algorithm in expectation or with high probability (w.h.p)\footnote{In this paper we take w.h.p to mean with probability of at least $1-1/n$. However, our results can be made to hold with probability $1-1/n^c$ for an arbitrarily large constant $c$ without affecting the asymptotic running times.} over the noise added to the input. Spielman and Teng show that the smoothed runtime of the simplex algorithm is \emph{polynomial}.
This approach is motivated by the fact that real-world inputs are not completely adversarial nor completely random, and can be seen as a middle ground between worst-case and average-case analysis. 

 \paragraph{Smoothed local Max-Cut} Let us now define the smoothed version of the local Max-Cut problem. Instead of assuming that $G$ has arbitrary edge weights, let the set of edge weights, $\set{w_e}_{e\in E}$, be independent random variables taking real values in $[-1, 1]$. Note that they need not be identically distributed. Each $w_e$ is chosen from a distribution whose probability density function is $f_e : [-1,1] \rightarrow [0, \phi]$. Where $\phi>1/2$ is a smoothing parameter, upper bounding the density of the distribution for any specific value. This parameter determines the strength of the adversary -- setting $\phi=1/2$ implies a uniform distribution of weights over $[-1, 1]$, while allowing $\phi \rightarrow \infty$ makes the adversary as strong as in the classical worst-case analysis paradigm. 
 
\paragraph{Smoothed analysis of FLIP}
The first to provide a quasi-polynomial bound for the smoothed runtime of FLIP on general graphs were \cite{EtscheidR17}. Prior to their work, it was only known that FLIP achieves a polynomial running time on graphs with logarithmic degree \cite{ElsasserT11}. It was first proved in \cite{AngelBPW17} that FLIP has smoothed \emph{polynomial} running time on \emph{complete} graphs. They provide a bound of $O(\phi^5 n^{15.1})$ w.h.p. 
The current state-of-the-art result for general graphs is due to \cite{ChenGVYZ20}, achieving a bound of $\phi n^{O(\sqrt{\log n})}$. The best-known result for complete graphs is due to \cite{BibakCC21}, where they provide a bound of $O(\phi n^{7.83})$ w.h.p. 


\paragraph{Our results} Proving that FLIP has a polynomial smoothed running time in general graphs is a major open problem in the field of smoothed analysis. A natural approach to attack this problem is to consider the problem on various graph classes of increasing complexity. Prior to our results, FLIP was only known to achieve a smoothed polynomial running time on two graph classes: complete graphs \cite{AngelBPW17, BibakCC21} and graphs with logarithmic maximum degree \cite{ElsasserT11}\footnote{Strictly speaking, the analysis provided in \cite{ElsasserT11} is for Gaussian perturbations. However, we believe that it can be easily generalized to general perturbations.}. In this paper, we focus on graphs with \emph{bounded arboricity}.

The arboricity of a graph is a measure of its sparsity. It is defined to be the minimum number of forests into which the edges can be partitioned. The arboricity is equal, up to a factor of 2, to the maximum average degree in any subgraph. Low arboricity graphs can be seen as being ``sparse everywhere".
Graphs with low arboricity are a very natural family of graphs which includes many important graph classes, such as minor closed graphs (e.g., planar graphs, bounded treewidth graphs), and randomly generate preferential attachment graphs (e.g., Barabasi-Albert \cite{barabasi1999emergence}). For an excellent exposition on the importance of low-arboricity graphs we refer the reader to \cite{EdenLR20, EdenRS20}.

We prove the following theorem:
\begin{restatable}{theorem}{mainthm}
\label{thm: mainthm}
Let $G$ be a graph with arboricity $\alpha$ where edge weights are independent random variables with density bounded by $\phi$, for any $\beta \in [2, n]$ it holds in expectation and w.h.p that FLIP terminates within 
$\phi n^{O(\frac{ \beta \alpha}{\log n} + \log_{\beta} \alpha)}$ iterations.
\end{restatable}
Setting $\beta = \log^{\delta} n ,\alpha = \log^{1-\delta} n$ we get the following running time:

\begin{align*}
    &\phi n^{O(\frac{ \beta \alpha}{\log n} + \log_{\beta} \alpha)} = \phi n^{O(\frac{ \log n}{\log n} + \log_{\log^\delta n} \log \log^{1 - \delta} n)} = \phi n^{O(1 + \frac{\log \log^{1 - \delta} n}{\log \log^\delta n})} 
    \\&= \phi n^{O(1 + \frac{(1 - \delta)\log \log n}{\delta \log \log n})} = \phi n^{O(1 + \frac{(1 - \delta)}{\delta})} = \phi n^{O(1/\delta)}
\end{align*}
The above implies that as long as $1/\delta = O(1)$ (i.e., $\delta$ is a constant), for any graph with arboricity $O(\log^{1-\delta} n)$ FLIP terminates in polynomial time in expectation and w.h.p.
\ifarxiv
\footnote{In a previous version of the manuscript we set $\alpha=\beta =O(\sqrt{\log n})$ which leads to sub-optimal results. This improved choice of $\alpha, \beta$ is due to Bruce Reed.}
\fi

Setting $\beta = O(1)$, for arbitrary values of $\alpha$ we get a running time of $\phi n^{O(\frac{\alpha}{\log n} + \log \alpha)}$. This improves over the best known running time for general graphs of $\phi n^{O(\sqrt{ \log n })}$ \cite{ChenGVYZ20} for $\alpha = o(\log^{1.5} n)$. Specifically, when $\alpha = O(\log n)$ we get a significantly faster running time of $\phi n^{O(\log \log n)}$.


\paragraph{Analysis outline} Our analysis draws inspiration from both \cite{ElsasserT11} and \cite{EtscheidR17}. In \cite{ElsasserT11} it is shown that w.h.p \emph{any} node movement (i.e., every step of FLIP) increases the cut weight by at least an additive $1/ \phi poly(n)$ factor, which immediately implies a running time of $\phi poly(n)$ (as the cut weight cannot exceed $n^2$). It is crucial for their analysis that the degree of \emph{all nodes} does not exceed $O(\log n)$ as otherwise we cannot get the guarantee to hold for all node movements. In \cite{EtscheidR17} longer sequences of movements (linear size) are considered. Specifically, they focus on consecutive movements of the same node in the sequence. They show that w.h.p any sufficiently long sequence of flips results in an additive improvement of at least $1/\phi n^{O(\log n)}$. 

Our analysis makes use of the bounded arboricity of the graph to construct a hierarchical partition\footnote{We would like to emphasize that the partition is only used for the sake of analysis and the algorithm is oblivious to it.} of the node set, $V=\cup_{i=1}^{\tau} V_i$ with the following property: every $v\in V_i$ has at most $O(\alpha)$ neighbors in $\cup_{j=i}^{\tau} V_j$. We use an argument similar to that of \cite{ElsasserT11} to guarantee that any 2 consecutive movements of $v \in V_i$, with no movements of any neighbors in $V_j, j < i$ in between them, significantly increase the cut weight (when $\alpha$ is sufficiently small). 
We then show that such ``good" sequences of movements must appear in any sufficiently long (super-polynomial) sequence of movements. Our analysis is clean and simple, and, to the best of our knowledge, it is the first to make use of structural properties of the input graph.

\section{FLIP on bounded arboricity graphs}


\paragraph{Vertex partition} We consider an input graph $G(V,E)$ with arboricity $\alpha$. Therefore, all induced subgraphs of $G$ have average degree at most $2\alpha$. We use this fact to create a partition of $V$. Let $\beta \in [2, n]$ be a parameter.
Starting with $G_1=G$, it has average degree at most $2\alpha$, therefore, it must be the case that at most $n/\beta$ vertices have degree larger than $2\beta \alpha$. Let $V_1$ be the vertices with degree $\leq 2\beta \alpha$. Let us consider the induced graph $G_2= G[V \setminus V_1]$. It has average degree at most $2\alpha$ and at most $\ceil{n/\beta}$ vertices. We take $V_2$ to be the nodes with degree $\leq 2\beta \alpha$ in $G_2$. Generally, we consider $G_i = G[V \setminus \cup_{j=1}^{i-1} V_j]$ and take $V_i$ to be the vertices with degree $\leq 2\beta \alpha$ in $G_i$. As every such step decreases the number of nodes in $G_i$ by a $\beta$ factor this process terminates in $\ceil{\log_\beta n}$ (going forward we omit the ceiling for ease of notation). This partitions $V$ into $V_1,\dots,V_{\log_\beta n}$. 

Let us denote by $N(v)$ the set of all neighbors of $v$ in $G$ and let $E(v)$ be the set of all edges adjacent to $v$. Let $V_i^+ = \cup_{j=i}^{\log_\beta n} V_j$ and for $v\in V_i$ let $N_*(v) = N(v) \cap V_i^+$. Similarly, let $E_*(v) = \set{(v,u) \in E(v) \mid u\in N_*(v)}$. Note that it holds that $\forall v\in V, \size{N_*(v)} \leq 2\beta \alpha$.


\paragraph{Good movements} Using the above decomposition we can now characterize specific ``good" sequences of movements that always lead to a significant improvement in the cut weight.

Let us consider some node $v\in V$, any movement of $v$ during the execution of FLIP results in an increase of the form $\sum_{ e \in E(v)} \lambda_e w_e$ where $\lambda_e \in \set{-1, 1}$. As long as $\set{w_e}_{e\in E(v)}$ are independent, it is known that for every assignment of values to $\set{\lambda_e}_{e\in E(v)}$ it holds that $Pr[\sum_{v\in e} \lambda_e w_e \in [0,\eps] ]\leq \epsilon \phi$ (i.e., the linear combination also has density bounded by $\phi$), this holds even if $\set{\lambda_e}_{e\in E(v)}$ take values in $\mathbb{Z}$ \cite{AngelBPW17}. If we could show that for any configuration of $\set{\lambda_e}_{e\in E(v)}$ it holds that $\sum_{e\in E(v)} \lambda_e w_e \notin [0,\eps]$ then we could guarantee that any movement of $v$ during the execution of FLIP results in an increase of at least $\epsilon$ to the cut weight, \emph{regardless} of the current cut configuration. 

Unfortunately, to achieve such a result, we must union bound over $2^{d(v)}$ possible configurations. This is the approach taken by \cite{ElsasserT11} to achieve their results for graphs with logarithmic maximum degree. However, for graphs with bounded arboricity, we need a new approach.

Similar to the analysis of \cite{EtscheidR17} we are interested in two consecutive movements of some node. However, we make use of our decomposition and focus on movements of
$v\in V_i$ that happen \emph{before} any nodes in $N(v) \setminus N_*(v)$ move. Summing the two increases, we get $\sum_{e\in E_*(v)} \lambda_e w_e$, where $\lambda_e \in \set{-2,0,2}$. We provide an illustration in Figure~\ref{fig:flip1}.
\begin{figure}[htbp] 
\label{fig: flip}
    \centering
    \includegraphics[width=\linewidth]{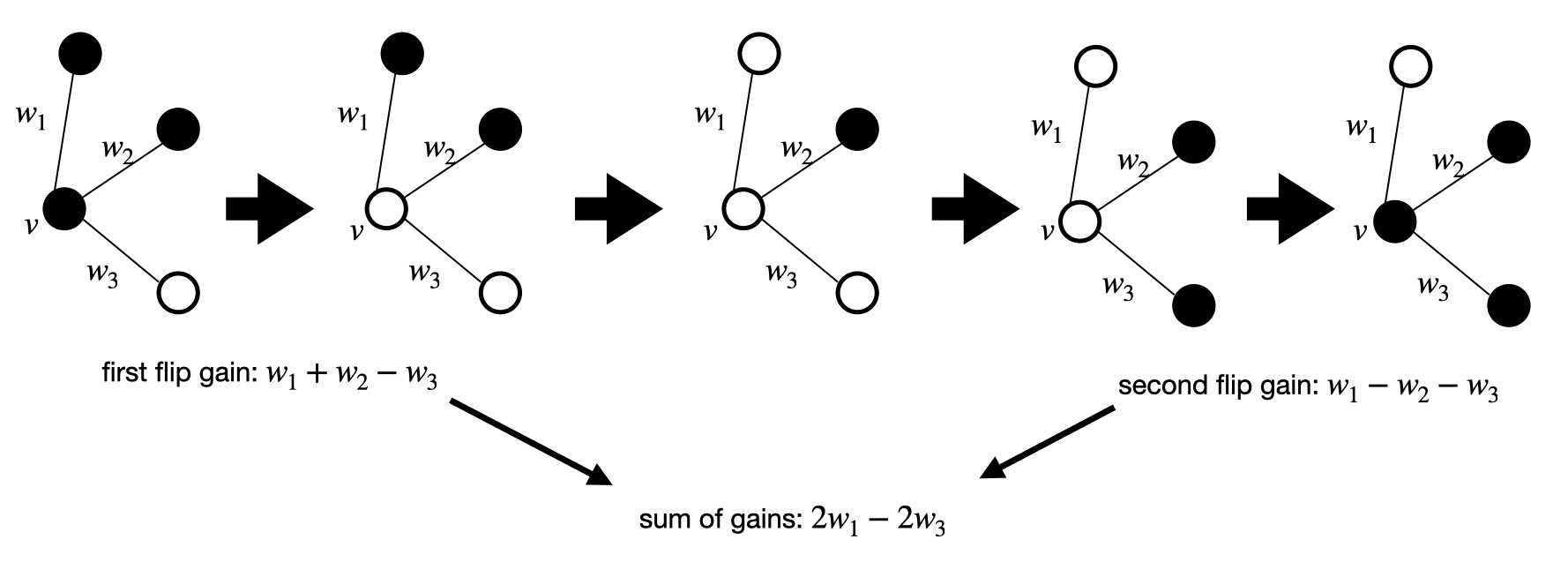}
    \caption{We consider the gain in the cut weight when $v$ is flipped for the first time plus the gain when it is flipped for the second time. We note that weights of edges to nodes that were flipped an even number of times between the two movement of $v$ get cancelled out and do not appear in the sum, while nodes that were flipped an odd number of times have a coefficient in $\set{-2,2}$.}
    \label{fig:flip1}
\end{figure}
This approach has the benefit of only considering linear combinations of at most $2\beta \alpha$ elements (instead of $d(v)$). This means that when applying a union bound it is sufficient to sum over only $3^{2\beta \alpha}$ configurations. We state the following lemma:
\begin{lemma}
\label{lem: inc guarantee}
For any $c\geq 1$, with probability at least $1- 1/n^c$ for all cut configurations and all nodes, for every two consecutive moves of $v\in V_i$ during the execution of FLIP such that no node in $N(v) \setminus N_*(v)$ moved in between them, the weight of the cut increases by at least $\phi^{-1}3^{-2\beta \alpha} n^{-c}$. 
\end{lemma}
\begin{proof}
We know that for every assignment of values in $\set{-2,0,2}$ to $\set{\lambda_e}_{e\in E_*(v)}$ it holds that $Pr[\sum_{e\in E_*(v)} \lambda_e w_e \in [0,\eps] ]\leq \epsilon \phi$. Setting $\epsilon = \phi^{-1}3^{-2\beta \alpha} n^{-c}$ and taking a union bound over all nodes and all possible $3^{2\beta \alpha}$ values of $\set{\lambda_e}_{e\in E_*(v)}$ we get that with probability at least $1-1/n^c$ for all nodes and all possible linear combinations with coefficients in $\set{-2,0,2}$ it holds that $\sum_{e\in E_*(v)} \lambda_e w_e \notin [0, \phi^{-1}3^{-2\beta \alpha} n^{-c}]$. 

\end{proof}

It is also worth noting that as we only consider linear combinations of the form $\sum_{e\in E_*(v)} \lambda_e w_e$ in the above Lemma, our results still hold under a weaker independence guarantee. It is sufficient that only $\set{w_e}_{e\in E_*(v)}$ are independent for every $v\in V$. 


\paragraph{Prevalence of good movements} Lemma~\ref{lem: inc guarantee} guarantees the existence of certain types of good movement sequences (parameterized by $c$). However, it is not clear how often these sequences must occur for executions of FLIP starting from an arbitrary cut configuration. We show that after a sufficiently long sequence of steps of FLIP the cut weight must increase by at least $\phi^{-1}3^{-2\beta \alpha} n^{-c}$. We state the following lemma.
\begin{lemma}
\label{lem: termination time}
    For any $c\geq 1$, with probability at least $1-1/n^c$, FLIP terminates within $\tau_c=n^{c}\cdot n^{O (\frac{\beta \alpha}{\log n} + \log_\beta \alpha)}$ iterations.
\end{lemma}
\begin{proof}

Starting from an arbitrary cut configuration, we note that if some node $v$ moves twice and no node $u \in N(v)\setminus N_*(v)$ moves in between, due to Lemma~\ref{lem: inc guarantee}, we are done. Therefore, we must bound the number of steps FLIP can make without allowing the above situation. Consider some execution of FLIP in which the above situation does not happen. It must be the case that if $v$ moves, then some node $u \in N(v)\setminus N_*(v)$ must move before $v$ moves again. 
Let us introduce some notation. Initially, all nodes are \emph{active}, after a node $v$ is flipped it becomes \emph{inactive} until a node $u \in N(v)\setminus N_*(v)$ is flipped (i.e., when $u$ is flipped, all nodes in $N_*(u)$ become active). Note that in a sequence without good movements, no inactive node will be flipped.
We show that eventually all nodes must become inactive.

Let us denote by $a_i(t)$ the number of active nodes in $V_i$ after the $t$-th step of the algorithm. 
Let us define the following potential function 
\[
\varphi(t) = \sum_{i=1}^{\log_\beta n} a_i(t) (\log_\beta n - i) (2\beta \alpha)^{(\log_\beta n) - i}
\]
Say that a node $v \in V_i$ is flipped at step $t$. Then $a_i$ is decreased by 1, and at most $2\beta \alpha$ other $a_j$'s ($j>i$) are incremented by 1. The largest possible increment is when all newly activated nodes are in $V_{i+1}$. This allows us to bound the change in the potential as follows:
\[
(\log_\beta n - (i+1)) 2\beta \alpha \cdot (2\beta \alpha)^{(\log_\beta n) - (i+1)}  - (\log_\beta n - i) (2\beta \alpha)^{(\log_\beta n) - i} = -\alpha^{(\log_\beta n) - i} < -1
\]

When the potential is 0 we conclude that there are no active nodes and either the algorithm is done, or the next flip will increase the cut weight significantly. As initially all nodes are active, the potential is upper bounded by:
\[
\sum_{i=1}^{\log_\beta n} \size{V_i} (\log_\beta n - i) (2\beta \alpha)^{(\log_\beta n) - i} \leq 
\sum_{i=1}^{\log_\beta n} \size{V_i} (\log_\beta n ) (2\beta \alpha)^{\log_\beta n} = (2\beta \alpha)^{\log_\beta n} n\log_\beta n = n^{O(\log_\beta \alpha)}
\]
We conclude that in every $n^{O(\log_\beta \alpha)}$ steps of FLIP we are guaranteed that the cut weight increases by at least $\phi^{-1}3^{-2\beta \alpha} n^{-c} = \phi^{-1}n^{-c} \cdot n^{-\Omega (\frac{\beta \alpha}{\log n})}$. As the cut weight cannot exceed $n^2$, the algorithm terminates within 
\[
\tau_c =\phi n^{-c} \cdot  n^{O (\frac{\beta \alpha}{\log n})} \cdot n^{O(\log_\beta \alpha)} = \phi n^{-c} \cdot n^{O (\frac{\beta \alpha}{\log n} + \log_\beta \alpha)}
\] 
steps with probability at least $1-1/n^c$. 
\end{proof}

Using the above, we state our main theorem.

\mainthm*
\begin{proof}
    The guarantee w.h.p follows directly from Lemma~\ref{lem: termination time} by setting $c$ to be any constant larger than 1. Let us now bound the expected running time. 
    
For ease of notation let us define $\tau_0 = 0$. Let $T$ be the random variable denoting the smoothed running time of FLIP. Using the law of total expectation we may express $E[T]$ as a sum of terms of the form $E[T \mid T \in [\tau_{c-1},\tau_c)]\cdot Pr[T \in [\tau_{c-1},\tau_c)]$. As the cut weight strictly increases throughout the execution of FLIP, and there are at most $2^n$ cut configurations, it holds that $T\leq2^n$ (with probability 1). Therefore, it is sufficient to sum the conditional expectations only up to $c=n$ (this is quite loose, but will not affect the asymptotics of the final solution). We get:
\begin{align*}
    E[T] = \sum_{c=1}^{n} E[T \mid T \in [\tau_{c-1},\tau_c)]\cdot Pr[T \in [\tau_{c-1},\tau_c)] \leq \sum_{c=1}^{n} \frac{\tau_c}{ n^{c-1}} = \sum_{c=1}^{n} \phi n^{O (\frac{\beta \alpha}{\log n} + \log_\beta \alpha)} = \phi n^{O (\frac{\beta \alpha}{\log n} + \log_\beta \alpha)} 
\end{align*}
This completes the proof.
\end{proof}

\ifarxiv
\paragraph{Acknowledgements} The author would like to thank Bruce Reed and Yuichi Sudo for helpful discussions and feedback.
\fi

\bibliographystyle{alpha}
\bibliography{paper}

\newcommand{\etalchar}[1]{$^{#1}$}
\begin{thebibliography}{ABPW17}

\bibitem[ABPW17]{AngelBPW17}
Omer Angel, S{\'{e}}bastien Bubeck, Yuval Peres, and Fan Wei.
\newblock Local max-cut in smoothed polynomial time.
\newblock In {\em {STOC}}, pages 429--437. {ACM}, 2017.

\bibitem[BA99]{barabasi1999emergence}
Albert-L{\'a}szl{\'o} Barab{\'a}si and R{\'e}ka Albert.
\newblock Emergence of scaling in random networks.
\newblock {\em science}, 286(5439):509--512, 1999.

\bibitem[BCC21]{BibakCC21}
Ali Bibak, Charles Carlson, and Karthekeyan Chandrasekaran.
\newblock Improving the smoothed complexity of {FLIP} for max cut problems.
\newblock {\em {ACM} Trans. Algorithms}, 17(3):19:1--19:38, 2021.

\bibitem[CGV{\etalchar{+}}20]{ChenGVYZ20}
Xi~Chen, Chenghao Guo, Emmanouil{-}Vasileios Vlatakis{-}Gkaragkounis, Mihalis
  Yannakakis, and Xinzhi Zhang.
\newblock Smoothed complexity of local max-cut and binary max-csp.
\newblock In {\em {STOC}}, pages 1052--1065. {ACM}, 2020.

\bibitem[ELR20]{EdenLR20}
Talya Eden, Reut Levi, and Dana Ron.
\newblock Testing bounded arboricity.
\newblock {\em {ACM} Trans. Algorithms}, 16(2):18:1--18:22, 2020.

\bibitem[ER17]{EtscheidR17}
Michael Etscheid and Heiko R{\"{o}}glin.
\newblock Smoothed analysis of local search for the maximum-cut problem.
\newblock {\em {ACM} Trans. Algorithms}, 13(2):25:1--25:12, 2017.

\bibitem[ERS20]{EdenRS20}
Talya Eden, Dana Ron, and C.~Seshadhri.
\newblock Faster sublinear approximation of the number of \emph{k}-cliques in
  low-arboricity graphs.
\newblock In {\em {SODA}}, pages 1467--1478. {SIAM}, 2020.

\bibitem[ET11]{ElsasserT11}
Robert Els{\"{a}}sser and Tobias Tscheuschner.
\newblock Settling the complexity of local max-cut (almost) completely.
\newblock In {\em {ICALP} {(1)}}, volume 6755 of {\em Lecture Notes in Computer
  Science}, pages 171--182. Springer, 2011.

\bibitem[FPT04]{FabrikantPT04}
Alex Fabrikant, Christos~H. Papadimitriou, and Kunal Talwar.
\newblock The complexity of pure nash equilibria.
\newblock In {\em {STOC}}, pages 604--612. {ACM}, 2004.

\bibitem[Hop82]{hopfield1982neural}
John~J Hopfield.
\newblock Neural networks and physical systems with emergent collective
  computational abilities.
\newblock {\em Proceedings of the national academy of sciences},
  79(8):2554--2558, 1982.

\bibitem[JPY88]{JohnsonPY88}
David~S. Johnson, Christos~H. Papadimitriou, and Mihalis Yannakakis.
\newblock How easy is local search?
\newblock {\em J. Comput. Syst. Sci.}, 37(1):79--100, 1988.

\bibitem[ST04]{SpielmanT04}
Daniel~A. Spielman and Shang{-}Hua Teng.
\newblock Smoothed analysis of algorithms: Why the simplex algorithm usually
  takes polynomial time.
\newblock {\em J. {ACM}}, 51(3):385--463, 2004.

\bibitem[SY91]{SchafferY91}
Alejandro~A. Sch{\"{a}}ffer and Mihalis Yannakakis.
\newblock Simple local search problems that are hard to solve.
\newblock {\em {SIAM} J. Comput.}, 20(1):56--87, 1991.

\end{thebibliography}
\end{document}